\documentclass[11pt, onecolumn]{article}
\usepackage[top=1in, bottom=1in, left=1.25in, right=1.25in]{geometry}

\usepackage{amsfonts}
\usepackage[cmex10]{amsmath}
\usepackage{amsthm}
\usepackage{cite}
\usepackage{color,soul}
\usepackage[pdftex]{graphicx}

\theoremstyle{definition}
\newtheorem{definition}{Definition}

\theoremstyle{definition}

\theoremstyle{definition}
\newtheorem{lemma}{Lemma}

\theoremstyle{definition}
\newtheorem{theorem}{Theorem}

\theoremstyle{remark}
\newtheorem{remark}{Remark}

\newcommand{\red}{\textcolor[rgb]{0.00,0.00,0.00}}

\begin{document}

\title{On the Null Space Constant for $\ell_p$ Minimization}

\author{Laming~Chen and~Yuantao~Gu 
\thanks{The authors are with the Department of Electronic Engineering, Tsinghua University, Beijing 100084, China. The corresponding author of this work is Yuantao Gu (e-mail: gyt@tsinghua.edu.cn).}}

\date{Submitted December 15, 2014, revised March 2, 2015}

\maketitle

\begin{abstract}
The literature on sparse recovery often adopts the $\ell_p$ ``norm'' $(p\in[0,1])$ as the penalty to induce sparsity of the signal satisfying an underdetermined linear system. The performance of the corresponding $\ell_p$ minimization problem can be characterized by its null space constant. In spite of the NP-hardness of computing the constant, its properties can still help in illustrating the performance of $\ell_p$ minimization. In this letter, we show the strict increase of the null space constant in the sparsity level $k$ and its continuity in the exponent $p$. We also indicate that the constant is strictly increasing in $p$ with probability $1$ when the sensing matrix ${\bf A}$ is randomly generated. Finally, we show how these properties can help in demonstrating the performance of $\ell_p$ minimization, mainly in the relationship between the the exponent $p$ and the sparsity level $k$.

\textbf{Keywords:} Sparse recovery, null space constant, $\ell_p$ minimization, monotonicity, continuity.
\end{abstract}

\section{Introduction}
\label{sec:introduction}

An important problem that often arises in signal processing, machine learning, and statistics is sparse recovery \cite{Baraniuk2007,Wright2009,Meinshausen2009}. It is in general formulated in the standard form
\begin{equation}\label{l0min}
    \underset{\bf x}{\operatorname{argmin}}\red{\|{\bf x}\|_0}\ \ \textrm{subject to}\ \ {\bf Ax}={\bf y}
\end{equation}
where the sensing matrix ${\bf A}\in\mathbb{R}^{M\times N}$ has more columns than rows and the $\ell_0$ ``norm'' $\red{\|{\bf x}\|_0}$ denotes the number of nonzero entries of the vector $\bf x$. The combinatorial optimization \eqref{l0min} is NP-hard and therefore \red{cannot} be solved efficiently \cite{Natarajan1995}. A standard method to solve this problem is by relaxing the non-convex discontinuous $\ell_0$ ``norm'' to the convex $\ell_1$ norm \cite{Candes2006}, i.e.,
\begin{equation}\label{l1min}
    \underset{\bf x}{\operatorname{argmin}}\|{\bf x}\|_1\ \ \textrm{subject to}\ \ {\bf Ax}={\bf y}.
\end{equation}
It is theoretically proved that under some certain \red{conditions \cite{Candes2006,Donoho2006}}, the optimum solution of \eqref{l1min} is identical to that of \eqref{l0min}.

Some works try to bridge the gap between $\ell_0$ ``norm'' and $\ell_1$ norm by non-convex but continuous $\ell_p$ ``norm'' \red{$(0<p<1)$ \cite{Chartrand2007,Saab2008,Foucart2009,Foucartnote2009}}, and consider the $\ell_p$ minimization problem
\begin{equation}\label{lpmin}
    \underset{\bf x}{\operatorname{argmin}}\|{\bf x}\|_p^p\ \ \textrm{subject to}\ \ {\bf Ax}={\bf y}
\end{equation}
where $\|{\bf x}\|_p^p=\sum_{i=1}^N|x_i|^p$. \red{Though finding the global optimal solution of $\ell_p$ minimization is still NP-hard, computing a local minimizer can be done in polynomial time \cite{Ge2011}.} The global optimality of \eqref{lpmin} has been studied and various conditions have been derived, for example, those based on \red{restricted isometry property \cite{Chartrand2007,Saab2008,Foucart2009,Sun2012}} and null space property \cite{Gribonval2007,Foucartnote2009}. Among them, a necessary and sufficient condition is based on the null space property and its constant \cite{Gribonval2007,Foucartnote2009,Chen2014}.

\begin{definition}\label{def_nsp}
\red{For any $0\le p\le1$, define} null space constant $\gamma(\ell_p,{\bf A},k)$ as the smallest quantity such that
\begin{equation}
    \sum_{i\in S}|z_i|^p\le\gamma(\ell_p,{\bf A},k)\sum_{i\not\in S}|z_i|^p
\end{equation}
holds for any set $S\subset\{1,2,\ldots,N\}$ with $\#S\le k$ and for any vector ${\bf z}\in \mathcal{N}({\bf A})$ which denotes the null space of ${\bf A}$.
\end{definition}

It has been shown that for any $p\in[0,1]$, $\gamma(\ell_p,{\bf A},k)<1$ is a necessary and sufficient condition such that for any $k$-sparse ${\bf x}^*$ and ${\bf y}={\bf Ax}^*$, ${\bf x}^*$ is the unique solution of $\ell_p$ minimization \cite{Foucartnote2009}. Therefore, $\gamma(\ell_p,{\bf A},k)$ is a tight quantity in indicating the performance of $\ell_p$ minimization $(0\le p\le 1)$ in sparse recovery. However, it has been shown that calculating $\gamma(\ell_p,{\bf A},k)$ is in general NP-hard \cite{Tillmann2012}, which makes it difficult to check whether the condition is satisfied or violated. Despite this, properties of $\gamma(\ell_p,{\bf A},k)$ are of tremendous help in illustrating the performance of $\ell_p$ minimization, e.g., non-decrease of $\gamma(\ell_p,{\bf A},k)$ in $p\in[0,1]$ shows that if $\ell_p$ minimization guarantees successful recovery of all $k$-sparse signal and $0\le q\le p$, then $\ell_q$ minimization also does \cite{Foucartnote2009}.

In this letter, we give some new properties of the null space constant $\gamma(\ell_p,{\bf A},k)$. Specifically, we prove that $\gamma(\ell_p,{\bf A},k)$ is strictly increasing in $k$ and is continuous in $p$. For random sensing matrix ${\bf A}$, the non-decrease of $\gamma(\ell_p,{\bf A},k)$ in $p$ can be improved to strict increase with probability 1. \red{Based on them, the performance of $\ell_p$ minimization can be intuitively demonstrated and understood.}

\section{Main Contribution}
\label{sec:contribution}

\red{This section introduces some properties of null space constant $\gamma(\ell_p,{\bf A},k)$ $(0\le p\le 1)$.} We begin with a lemma about $\gamma(\ell_p,{\bf A},k)$ which will play \red{a} central role in the theoretical analysis. The spark of a matrix ${\bf A}$, denoted as $\mathrm{Spark}({\bf A})$ \cite{Donoho2003}, is the smallest number of columns from ${\bf A}$ that are linearly dependent.

\begin{lemma}\label{lem_nsp}
Suppose $\mathrm{Spark}({\bf A})=L+1$. For $p\in[0,1]$,
\begin{itemize}
    \item[1)]
    $\gamma(\ell_p,{\bf A},k)$ is finite if and only if $k\le L$;
    \item[2)]
    For $k\le L$, there exist $S'\subset\{1,2,\ldots,N\}$ with $\#S'\le k$ and ${\bf z}'\in \mathcal{N}({\bf A})\setminus\{\bf 0\}$ such that
    \begin{equation}\label{equ10}
        \sum_{i\in S'}|z'_i|^p=\gamma(\ell_p,{\bf A},k)\sum_{i\not\in S'}|z'_i|^p
    \end{equation}
\end{itemize}
\end{lemma}

\begin{proof}
See Section~\ref{ssec:prf_lem_nsp}.
\end{proof}

First, we show the strict increase of $\gamma(\ell_p,{\bf A},k)$ in $k$.

\begin{theorem}\label{thm_nsp_k}
Suppose $\mathrm{Spark}({\bf A})=L+1$. Then for $p\in[0,1]$, $\gamma(\ell_p,{\bf A},k)$ is strictly increasing in $k$ when $k\le L$.
\end{theorem}

\begin{proof}
See Section~\ref{ssec:prf_thm_nsp_k}.
\end{proof}

\begin{remark}
For any $p\in[0,1]$, we can define a set $\mathcal{K}_p({\bf A})$ of all positive integers $k$ that every $k$-sparse ${\bf x}^*$ can be recovered as the unique solution of $\ell_p$ minimization \eqref{lpmin} with ${\bf y}={\bf Ax}^*$. According to Theorem~\ref{thm_nsp_k}, $\mathcal{K}_p({\bf A})$ contains successive integers starting from $1$ to some integer $k_p^*({\bf A})$ and is possibly empty.
\end{remark}

\begin{remark}
\red{If $\mathrm{Spark}({\bf A})=L+1$, then $k_0^*({\bf A})=\lfloor L/2\rfloor$ \cite{Donoho2003}}. Therefore, if $L\ge2$, $k_0^*({\bf A})\ge1$.
\end{remark}

\begin{remark}
For $\bf A$ with identical column norms, if $\mathrm{Spark}({\bf A})=L+1$ and $L\ge2$, then $k_1^*({\bf A})\ge1$. To show this, we only need to prove that $\gamma(\ell_1,{\bf A},1)<1$. First, for any $1\le i\le N$ and ${\bf z}\in\mathcal{N}({\bf A})\setminus\{\bf 0\}$, since ${\bf Az=0}$, $z_i{\bf a}_i=-\sum_{j\ne i}z_j{\bf a}_j$ where ${\bf a}_i$ is the $i$th column of $\bf A$. Since
\[
    |z_i|\cdot\|{\bf a}_i\|_2=\|z_i{\bf a}_i\|_2=\bigg\|\sum_{j\ne i}z_j{\bf a}_j\bigg\|_2\le\sum_{j\ne i}|z_j|\cdot\|{\bf a}_j\|_2
\]
with equality holds only when $z_j{\bf a}_j$ $(j\ne i)$ are all on the same ray, which cannot be true since $\mathrm{Spark}({\bf A})=L+1\ge 3$. Since $\bf A$ has identical column norms, $|z_i|<\sum_{j\ne i}|z_j|$ holds, which leads to $\gamma(\ell_1,{\bf A},1)<1$ because of Lemma~\ref{lem_nsp}.2).
\end{remark}

Now we turn to the properties of $\gamma(\ell_p,{\bf A},k)$ as a function of $p$. The following result reveals the continuity of $\gamma(\ell_p,{\bf A},k)$ in $p$.

\begin{theorem}\label{thm_nsp_p2}
Suppose $\mathrm{Spark}({\bf A})=L+1$. Then for $k\le L$, $\gamma(\ell_p,{\bf A},k)$ is a continuous function in $p\in[0,1]$.
\end{theorem}

\begin{proof}
See Section~\ref{ssec:prf_thm_nsp_p2}.
\end{proof}

\begin{remark}
Some works have discussed the equivalence of $\ell_0$ and $\ell_p$ minimizations. In \cite{Malioutov2004}, it is shown that the sufficient condition for the equivalence of these two minimization problems approaches the necessary and sufficient condition for the uniqueness of solutions of $\ell_0$ minimization. In \cite{Chartrand2007}, it is shown that for any $k$-sparse ${\bf x}^*$ and ${\bf y}={\bf Ax}^*$, if $\delta_{2k+1}<1$, then there is $p>0$ such that ${\bf x}^*$ is the unique solution of $\ell_p$ minimization. \red{This result is improved to $\delta_{2k}<1$ which is optimal since it is exactly the necessary and sufficient condition for ${\bf x}^*$ being the unique solution of $\ell_0$ minimization \cite{Sun2012}.} \cite{Fung2011} shows the equivalence of the $\ell_0$- and the $\ell_p$-norm minimization problem for sufficiently small $p$. According to Theorem~\ref{thm_nsp_p2}, we can also justify this result: For any $k$-sparse ${\bf x}^*$ and ${\bf y}={\bf Ax}^*$, if $\gamma(\ell_0,{\bf A},k)<1$, then there is $p>0$ such that $\gamma(\ell_p,{\bf A},k)<1$ and ${\bf x}^*$ is the unique solution of $\ell_p$ minimization.
\end{remark}

\begin{remark}
In \cite{Foucartnote2009}, the author defines a set $\mathcal{P}_k({\bf A})$ of reconstruction exponents, that is the set of all exponents $0<p\le 1$ for which every $k$-sparse ${\bf x}^*$ is recovered as the unique solution of $\ell_p$ minimization with ${\bf y}={\bf Ax}^*$. It is shown that $\mathcal{P}_k({\bf A})$ is a (possibly empty) open interval $(0,p_k^*({\bf A}))$ \cite{Foucartnote2009}. This result can be easily shown by Theorem~\ref{thm_nsp_p2}. Since $\gamma(\ell_p,{\bf A},k)$ is a non-decreasing \cite{Gribonval2007} continuous function in $p\in[0,1]$, the inverse image of the open interval $(-\infty,1)$ is also an open interval of $[0,1]$. Therefore, the requirement that $\gamma(\ell_p,{\bf A},k)<1$ is equivalent to $p\in[0,p_k^*({\bf A}))$.
\end{remark}

\begin{figure}[t]
\begin{center}
\includegraphics[width=0.6\textwidth]{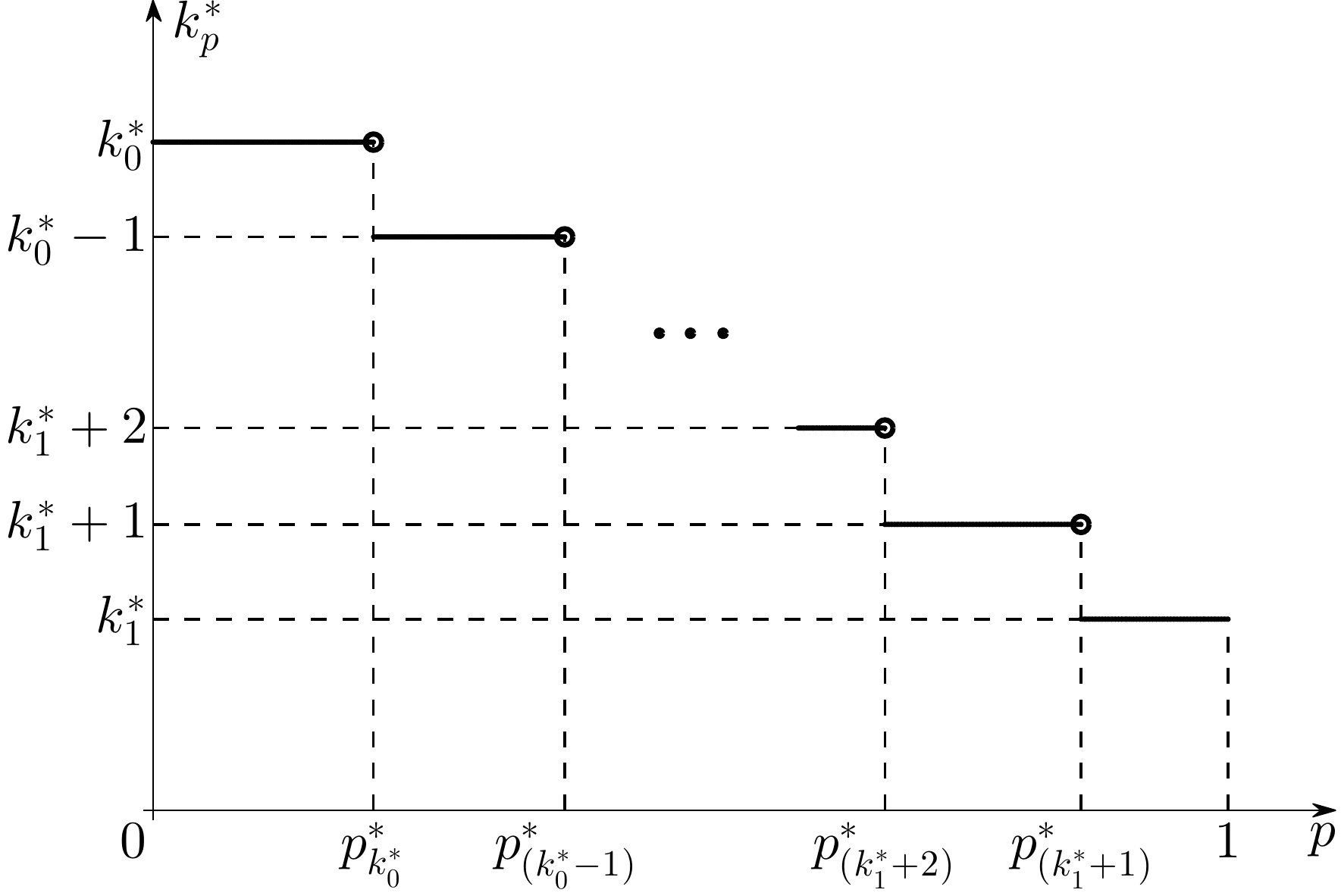}
\caption{The figure shows $k_p^*({\bf A})$ as a function of $p$, where the argument $\bf A$ is omitted for concision.}\label{fig:fig1}
\end{center}
\end{figure}

\begin{remark}
For any $\bf A$, we can plot $k_p^*({\bf A})$ as a function of $p$, as shown in Fig.~\ref{fig:fig1}. For concision, we omit the argument $\bf A$ in the figure. It is obvious that $k_p^*({\bf A})$ is a step function decreasing from $k_0^*({\bf A})$ to $k_1^*({\bf A})$. Three facts needs to be pointed out. First, $k_p^*({\bf A})$ is right-continuous, which is an easy consequence of Theorem~\ref{thm_nsp_p2}. Second, the points $(p_0,k_0)$ corresponding to the hollow circles in Fig.~\ref{fig:fig1} satisfy $\gamma(\ell_{p_0},{\bf A},k_0)=1$. Third, for the $p$-axis $p_0$ of the points of discontinuity, the one-sided limits satisfy $\lim_{p\rightarrow p_0^-}k_p^*({\bf A})-\lim_{p\rightarrow p_0^+}k_p^*({\bf A})=1$. This can be proved by Theorem~\ref{thm_nsp_k} that if $\gamma(\ell_{p_0},{\bf A},k_0)=1$, then $\gamma(\ell_{p_0},{\bf A},k_0-1)<1$.
\end{remark}

Finally, we introduce an important property of $\gamma(\ell_p,{\bf A},k)$ as a function of $p$ with regard to random matrix ${\bf A}$.

\begin{theorem}\label{thm_nsp_p3}
Suppose the entries of ${\bf A}\in\mathbb{R}^{M\times N}$ are i.i.d. and satisfy a continuous probability distribution. Then for $k\le M$, $\gamma(\ell_p,{\bf A},k)$ is strictly increasing in $p\in[0,1]$ with probability one.
\end{theorem}

\begin{proof}
See Section~\ref{ssec:prf_thm_nsp_p3}.
\end{proof}

\begin{remark}
It needs to be noted that there exists ${\bf A}$ such that $\gamma(\ell_p,{\bf A},k)$ is a constant number for all $p\in[0,1]$. For example, for
\begin{equation}
{\bf A}=\frac{1}{\sqrt{2}}\begin{bmatrix} 1 & 1 \\ 1 & 1 \end{bmatrix},
\end{equation}
$\mathrm{Spark}({\bf A})=2$. Since $\mathcal{N}({\bf A})=\textrm{span}([1, -1]^{T})$, \red{it is} easy to check that for all $p\in[0,1]$, $\gamma(\ell_p,{\bf A},1)=1$.
\end{remark}

\begin{figure}[t]
\begin{center}
\includegraphics[width=0.6\textwidth]{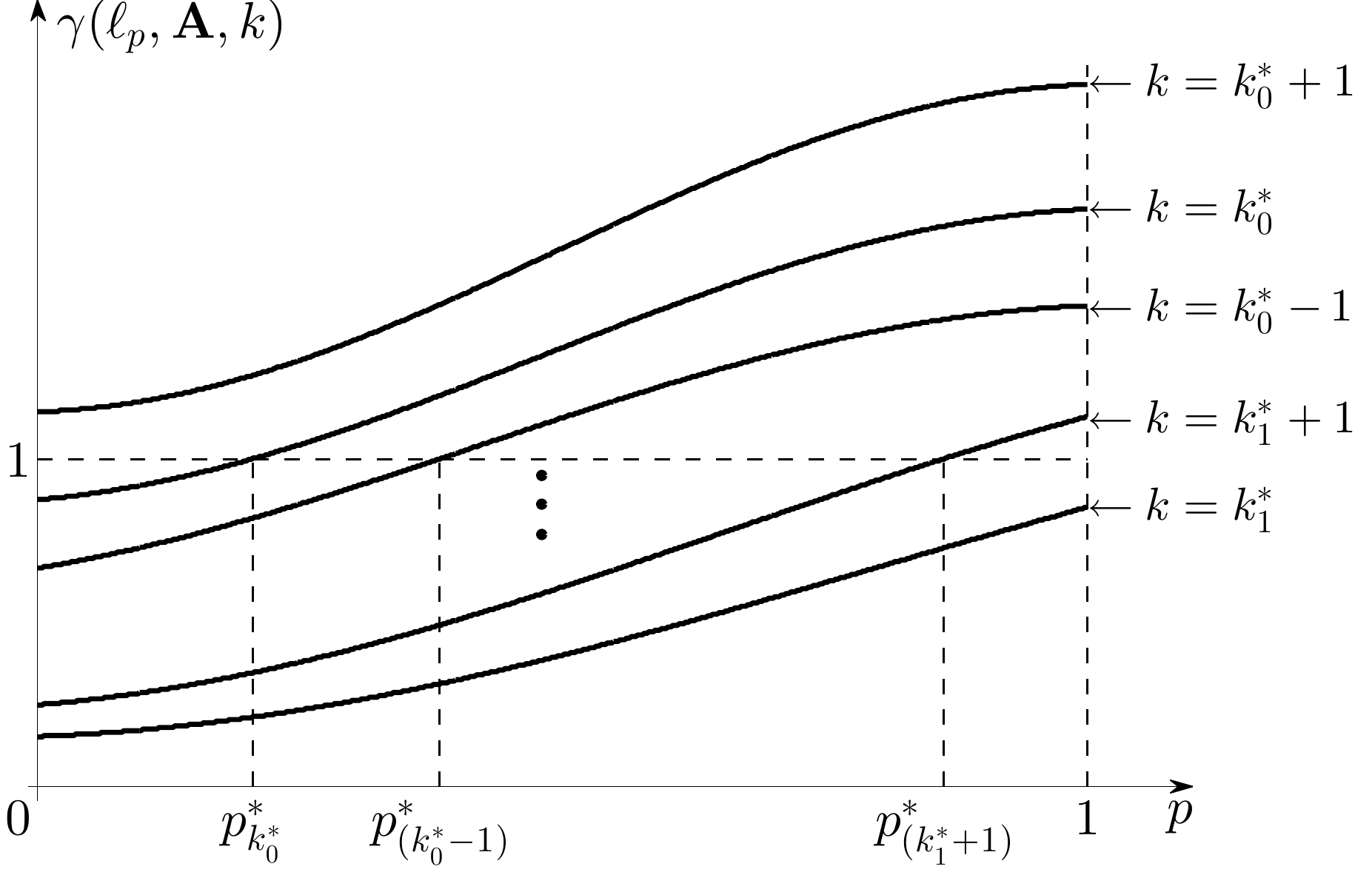}
\caption{This figure shows a diagrammatic sketch of $\gamma(\ell_p,{\bf A},k)$ as a function of $p$ for different $k$ when $\bf A$ is a random matrix.}\label{fig:fig2}
\end{center}
\end{figure}

\begin{remark}
To sum up, we can schematically show $\gamma(\ell_p,{\bf A},k)$ as a function of $p$ for different $k$ in Fig.~\ref{fig:fig2}. According to Theorem~\ref{thm_nsp_k}, these curves are strictly in order without intersections. Theorem~\ref{thm_nsp_p2} reveals that $\gamma(\ell_p,{\bf A},k)$ is continuous in $p$. For a random matrix $\bf A$ with i.i.d. entries satisfying a continuous probability distribution, $\gamma(\ell_p,{\bf A},k)$ is strictly increasing in $p$ with probability 1 by Theorem~\ref{thm_nsp_p3}. According to the definition of $k_p^*({\bf A})$, the curves intersecting $\gamma(\ell_p,{\bf A},k)=1$ $(0\le p\le1)$ are those with $k_1^*({\bf A})+1\le k\le k_0^*({\bf A})$. According to the definition of $p_k^*({\bf A})$, the $p$-axis of these intersections are $p_{k_0^*}^*$, $p_{k_0^*-1}^*$, $\dots$, $p_{k_1^*+1}^*$ from left to right. Therefore, it is easy to derive Fig.~\ref{fig:fig1} based on Fig.~\ref{fig:fig2} when $\bf A$ is a random matrix.
\end{remark}

\section{Proofs}
\label{sec:proof}

\subsection{Proof of Lemma~\ref{lem_nsp}}
\label{ssec:prf_lem_nsp}

\begin{proof}

1) Since $\mathrm{Spark}({\bf A})=L+1$, $\mathcal{N}({\bf A})$ contains an $(L+1)$-sparse signal, and it is easy to show that for any $k\ge L+1$, $\gamma(\ell_p,{\bf A},k)=+\infty$ according to Definition~\ref{def_nsp}. Next we prove that for $k\le L$, $\gamma(\ell_p,{\bf A},k)$ is finite. Define
\begin{equation}
\theta(p,{\bf z},S)=\frac{\sum_{i\in S}|z_i|^p}{\sum_{i\not\in S}|z_i|^p}
\end{equation}
and $\mathcal{N}_1({\bf A})=\mathcal{N}({\bf A})\cap\{{\bf z}:\|{\bf z}\|_2=1\}$ which is a compact set. Then it is easy to see that the definition of null space constant is equivalent to
\begin{equation}\label{equ3}
    \gamma(\ell_p,{\bf A},k)=\max_{\#S\le k}\sup_{{\bf z}\in\mathcal{N}_1({\bf A})}\theta(p,{\bf z},S).
\end{equation}
If $\gamma(\ell_p,{\bf A},k)$ is not finite, then there exists $S'$ with $\#S'\le k$ such that $\sup_{{\bf z}\in\mathcal{N}_1({\bf A})}\theta(p,{\bf z},S')$ is not finite. Therefore, for any $n\in\mathbb{N}^+$, there exists ${\bf z}^{(n)}\in\mathcal{N}_1({\bf A})$ such that
\begin{equation}\label{equ1}
    \theta(p,{\bf z}^{(n)},S')\ge n.
\end{equation}
If $p=0$, since ${\bf z}^{(n)}$ is at least $(L+1)$-sparse, \red{it is} easy to see that $\theta(0,{\bf z}^{(n)},S')\le k$ holds for any $n\in\mathbb{N}^+$. This contradicts \eqref{equ1} when $n>k$. If $p\in(0,1]$, according to Lemma~4.5 in \cite{Foucartnote2009}, $\|{\bf z}^{(n)}\|_p\le N^{\frac{1}{p}-\frac{1}{2}}\|{\bf z}^{(n)}\|_2=N^{\frac{1}{p}-\frac{1}{2}}$, and \eqref{equ1} implies
\begin{equation}\label{equ2}
    \sum_{i\not\in S'}|z_i^{(n)}|^p\le\frac{N^{1-\frac{p}{2}}}{n+1}.
\end{equation}
Due to the compactness of $\mathcal{N}_1({\bf A})$, the sequence $\{{\bf z}^{(n)}\}_n$ has a convergent subsequence $\{{\bf z}^{(n_m)}\}_m$, and its limit ${\bf z}'$ also lies in $\mathcal{N}_1({\bf A})$. Then \eqref{equ2} implies $z'_i=0$ for $i\not\in S'$, i.e., $\mathcal{N}_1({\bf A})$ contains a $k$-sparse element ${\bf z}'$. This contradicts the assumption that $\mathrm{Spark}({\bf A})=L+1>k$.

2) If $p=0$, for any $S$ with $\#S\le k$ and any ${\bf z}\in\mathcal{N}({\bf A})\setminus\{{\bf 0}\}$, it holds that
\begin{equation}
    \theta(0,{\bf z},S)\le\frac{k}{L+1-k}.
\end{equation}
On the other hand, since $\mathrm{Spark}({\bf A})=L+1$, $\mathcal{N}({\bf A})$ contains an $(L+1)$-sparse signal ${\bf z}'$ with $T$ as its support set. For any $S'\subset T$ with $\#S'=k$, $\theta(0,{\bf z}',S')=k/(L+1-k)$, and therefore \eqref{equ10} holds.

If $p\in(0,1]$, recalling the equivalent definition \eqref{equ3}, there exists $S'$ with $\#S'\le k$ such that
\begin{equation}
    \gamma(\ell_p,{\bf A},k)=\sup_{{\bf z}\in\mathcal{N}_1({\bf A})}\theta(p,{\bf z},S').
\end{equation}
Since $\mathcal{N}_1({\bf A})$ is compact and the function $\theta(p,{\bf z},S')$ is continuous in $\bf z$ on $\mathcal{N}_1({\bf A})$, it is easy to show that there exists ${\bf z}'\in\mathcal{N}_1({\bf A})$ such that $\gamma(\ell_p,{\bf A},k)=\theta(p,{\bf z}',S')$.
\end{proof}

\subsection{Proof of Theorem~\ref{thm_nsp_k}}
\label{ssec:prf_thm_nsp_k}

\begin{proof}
We prove that when $p\in[0,1]$ and $2\le k\le L$,
\begin{equation}\label{equ4}
    \gamma(\ell_p,{\bf A},k-1)<\gamma(\ell_p,{\bf A},k).
\end{equation}
According to Lemma~\ref{lem_nsp}.2), there exist $S'$ with $\#S'\le k-1$ and ${\bf z}'\in \mathcal{N}_1({\bf A})$ such that
\begin{equation}\label{equ17}
\gamma(\ell_p,{\bf A},k-1)=\theta(p,{\bf z}',S').
\end{equation}
Since ${\bf z}'$ is at least $(L+1)$-sparse, there exists an index $s'\in\{1,2,\ldots,N\}\setminus S'$ such that $z'_{s'}\ne0$. Let $S''=S'\cup\{s'\}$, then
\begin{equation}
    \sum_{i\in S'}|z'_i|^p<\sum_{i\in S''}|z'_i|^p,\ \ \ \ \sum_{i\not\in S'}|z'_i|^p>\sum_{i\not\in S''}|z'_i|^p>0
\end{equation}
and hence
\begin{equation}\label{equ6}
    \theta(p,{\bf z}',S')<\theta(p,{\bf z}',S'').
\end{equation}
Recalling \eqref{equ17} and the equivalent definition \eqref{equ3}, we can get \eqref{equ4} and complete the proof.
\end{proof}

\subsection{Proof of Theorem~\ref{thm_nsp_p2}}
\label{ssec:prf_thm_nsp_p2}

\begin{proof}
According to Theorem~5 in \cite{Gribonval2007}, $\gamma(\ell_p,{\bf A},k)$ is non-decreasing in $p\in[0,1]$ and therefore can only have jump discontinuities. We show this is impossible by two steps.

First, for any $p\in(0,1]$, we prove the one-sided limit from the negative direction satisfies
\begin{equation}\label{equ9}
    L^-:=\lim_{q\rightarrow p^-}\gamma(\ell_q,{\bf A},k)=\gamma(\ell_p,{\bf A},k).
\end{equation}
According to Lemma~\ref{lem_nsp}.2), there exist $S'$ with $\#S'\le k$ and ${\bf z}'\in\mathcal{N}_1({\bf A})$ satisfying
\begin{equation}
    \gamma(\ell_p,{\bf A},k)=\theta(p,{\bf z}',S').
\end{equation}
According to the definition of $\theta(p,{\bf z},S)$, it is easy to show that
\begin{equation}
    \lim_{q\rightarrow p^-}\theta(q,{\bf z}',S')=\theta(p,{\bf z}',S'),
\end{equation}
and then \eqref{equ9} holds obviously.

Second, for any $p\in[0,1)$, we prove the one-sided limit from the positive direction satisfies
\begin{equation}\label{equ11}
    L^+:=\lim_{q\rightarrow p^+}\gamma(\ell_q,{\bf A},k)=\gamma(\ell_p,{\bf A},k).
\end{equation}
Since $p<1$, there exists $N_0\in\mathbb{N}^+$ such that $p+N_0^{-1}\le1$. Then for $n\ge N_0$, Lemma~\ref{lem_nsp}.2) reveals that there exist $S^{(n)}$ with $\#S^{(n)}\le k$ and ${\bf z}^{(n)}\in\mathcal{N}_1({\bf A})$ such that
\begin{equation}\label{equ12}
    \gamma(\ell_{p+n^{-1}},{\bf A},k)=\theta(p+n^{-1},{\bf z}^{(n)},S^{(n)}).
\end{equation}
Since there are only finite different $S$ satisfying $\#S\le k$, there exists $S'$ with $\#S'\le k$ such that an infinite subsequence of $\{{\bf z}^{(n)}\}_n$ is associated with $S'$. Due to the compactness of $\mathcal{N}_1({\bf A})$, this subsequence has a convergent subsequence $\{{\bf z}^{(n_m)}\}_m$, and its limit ${\bf z}'$ also lies in $\mathcal{N}_1({\bf A})$. According to the definition of $\theta(p,{\bf z},S)$ and \eqref{equ12},
\begin{equation}
    \theta(p,{\bf z}',S')=\lim_{m\rightarrow+\infty}\theta(p+n_m^{-1},{\bf z}^{(n_m)},S')=L^+,
\end{equation}
and consequently $\gamma(\ell_p,{\bf A},k)\ge L^+$. Since $\gamma(\ell_p,{\bf A},k)$ is non-decreasing in $p$, $\gamma(\ell_p,{\bf A},k)\le L^+$ and \eqref{equ11} is proved.
\end{proof}

\subsection{Proof of Theorem~\ref{thm_nsp_p3}}
\label{ssec:prf_thm_nsp_p3}

\begin{proof}
First, we show that $\mathrm{Spark}({\bf A})=M+1$ with probability 1. Let $\mathcal{M}(M)$ denote the $M^2$-dimensional vector space of $M\times M$ real matrices. For any $0\le k\le M$, let $\mathcal{M}_k(M)$ denote the subset of $\mathcal{M}(M)$ consisting of matrices of rank $k$. It can be proved that $\mathcal{M}_k(M)$ is an embedded submanifold of dimension $k(2M-k)$ in $\mathcal{M}(M)$ \cite{Lee2012}. Consequently, for $M\times M$ matrices with i.i.d. entries drawn from a continuous distribution, the $M^2$-dimensional volume of the set of singular matrices $\bigcup_{k=0}^{M-1}\mathcal{M}_k(M)$ is zero. In other words, any $M$, or fewer, random vectors in $\mathbb{R}^M$ with i.i.d. entries drawn from a continuous distribution are linearly independent with probability 1. On the other hand, more than $M$ vectors in $\mathbb{R}^M$ are always linearly dependent. Therefore, $\mathrm{Spark}({\bf A})=M+1$ with probability 1.

Next, with the equivalent definition \eqref{equ3}, we prove that for $k\le M$ and $0\le p<q\le 1$,
\begin{equation}\label{equ13}
    \max_{\#S\le k}\sup_{{\bf z}\in\mathcal{N}_1({\bf A})}\theta(p,{\bf z},S)<\max_{\#S\le k}\sup_{{\bf z}\in\mathcal{N}_1({\bf A})}\theta(q,{\bf z},S)
\end{equation}
holds with probability 1. According to Lemma~\ref{lem_nsp}.2), there exist $S'$ with $\#S'\le k$ and ${\bf z}'\in\mathcal{N}_1({\bf A})$ such that
\begin{equation}\label{equ15}
    \theta(p,{\bf z}',S')=\max_{\#S\le k}\sup_{{\bf z}\in\mathcal{N}_1({\bf A})}\theta(p,{\bf z},S).
\end{equation}
Suppose ${\bf z}'$ has $N_*$ nonzero entries with $T$ as its support set, then $N_*\ge M+1$ with probability 1. \red{It is} obvious that $S'\subset T$, and for any $i\in S'$ and any $l\in T\setminus S'$, $|z'_i|\ge|z'_l|>0$. Since $p<q$, $|z'_i|^{q-p}\ge|z'_l|^{q-p}$ and therefore
\begin{equation}\label{equ7}
    |z'_i|^q|z'_l|^p\ge|z'_i|^p|z'_l|^q.
\end{equation}
Summing \eqref{equ7} with $i$ in $S'$ and $l$ in $T\setminus S'$, we can obtain
\begin{equation}\label{equ14}
    \sum_{i\in S'}|z'_i|^q\sum_{l\in T\setminus S'}|z'_l|^p\ge\sum_{i\in S'}|z'_i|^p\sum_{l\in T\setminus S'}|z'_l|^q
\end{equation}
which is equivalent to
\begin{equation}\label{equ8}
    \theta(p,{\bf z}',S')\le\theta(q,{\bf z}',S').
\end{equation}
Since $p<q$, it is easy to check that the equality in \eqref{equ8} holds only when $|z'_i|=|z'_l|$ for all $i\in S'$ and all $l\in T\setminus S'$, i.e., the nonzero entries of ${\bf z}'$ have the same magnitude. We prove that $\mathcal{N}_1({\bf A})$ contains such ${\bf z}'$ with probability 0, which together with \eqref{equ15} imply that
\begin{equation}
    \gamma(\ell_p,{\bf A},k)=\theta(p,{\bf z}',S')<\theta(q,{\bf z}',S')\le\gamma(\ell_q,{\bf A},k)
\end{equation}
holds with probability 1.

To this end, let $\mathcal{M}(M,N)$ denote the $MN$-dimensional vector space of $M\times N$ real matrices. For fixed ${\bf z}\in\mathbb{R}^N$ with $\|{\bf z}\|_2=1$, it can be easily shown that the subset
\begin{equation}
    \mathcal{M}_{\bf z}(M,N)=\{{\bf A}\in\mathcal{M}(M,N):{\bf Az}={\bf 0}\}
\end{equation}
is an $M(N-1)$-dimensional subspace in $\mathcal{M}(M,N)$. Therefore, for ${\bf A}\in\mathcal{M}(M,N)$ with i.i.d. entries drawn from a continuous probability distribution, $\mathcal{N}_1({\bf A})$ contains ${\bf z}$ with probability 0. In $\{{\bf z}\in\mathbb{R}^N:\|{\bf z}\|_2=1\}$, the number of vectors whose nonzero entries have the same magnitude is
\begin{equation}
    \sum_{i=1}^N{N\choose i}2^i=3^N-1
\end{equation}
which is a finite number. Therefore, with probability 0, $\mathcal{N}_1({\bf A})$ contains a vector ${\bf z}'$ which makes the equality in \eqref{equ8} hold. That is, $\gamma(\ell_p,{\bf A},k)$ is strictly increasing in $p\in[0,1]$ with probability 1.
\end{proof}

\section{Conclusion}
\label{sec:conclusion}

In characterizing the performance of $\ell_p$ minimization in sparse recovery, null space constant $\gamma(\ell_p,{\bf A},k)$ can be served as a necessary and sufficient condition for the perfect recovery of all $k$-sparse signals. This letter derives some basic properties of $\gamma(\ell_p,{\bf A},k)$ in $k$ and $p$. In particular, we show that $\gamma(\ell_p,{\bf A},k)$ is strictly increasing in $k$ and is continuous in $p$, meanwhile for random $\bf A$, the constant is strictly increasing in $p$ with probability 1. Possible future works include the properties of $\gamma(\ell_p,{\bf A},k)$ in $\bf A$, for example, the requirement of number of measurements $M$ to guarantee $\gamma(\ell_p,{\bf A},k)<1$ with high probability when $\bf A$ is randomly generated.

\end{document}